\documentclass[a4paper,UKenglish,cleveref, autoref, thm-restate]{lipics-v2019}

\usepackage{multirow}
\usepackage[ruled]{algorithm}
\usepackage[noend]{algorithmic}
\usepackage{amssymb, mathrsfs}
\usepackage{pdflscape}
\hideLIPIcs
\nolinenumbers

\title{An integer programming formulation using convex polygons for the convex partition problem} 

\titlerunning{An IP formulation using convex polygons for the convex partition problem} 

\author{Hadrien Cambazard}{Univ. Grenoble Alpes, CNRS, Grenoble INP, G-SCOP, F-38000 Grenoble, France
}{hadrien.cambazard@grenoble-inp.fr}{}{}

\author{Nicolas Catusse}{Univ. Grenoble Alpes, CNRS, Grenoble INP, G-SCOP, F-38000 Grenoble, France
}{nicolas.catusse@grenoble-inp.fr}{}{Supported in part by the ANR Project DISTANCIA (ANR-17-CE40-0015) operatedby the French National Research Agency (ANR).}

\authorrunning{H. Cambazard and N. Catusse} 

\Copyright{H. Cambazard and N. Catusse} 



\category{} 

\relatedversion{} 

\supplement{}

\acknowledgements{}

\EventEditors{}
\EventNoEds{2}
\EventLongTitle{The 37th International Symposium on Computational Geometry (SOCG 2021)}
\EventShortTitle{SOCG 2021}
\EventAcronym{SOCG}
\EventYear{2021}
\EventDate{June 7--11, 2021}
\EventLocation{Buffalo, USA}
\EventLogo{}
\SeriesVolume{}
\ArticleNo{}

\begin{document}

\maketitle

\begin{abstract}
A convex partition of a point set $P$ in the plane is a planar partition of the convex hull of $P$ with empty convex polygons or internal faces whose extreme points belong to $P$. In a convex partition, the union of the internal faces give the convex hull of $P$ and the interiors of the polygons are pairwise disjoint. Moreover, no polygon is allowed to contain a point of $P$ in its interior. The problem is to find a convex partition based on the minimum number of internal faces. The problem has been shown to be NP-Hard and was recently used in the CG:SHOP Challenge 2020. We propose a new integer linear programming (IP) formulation that considerably improves over the existing one. It relies on the representation of faces as opposed to segments and points. A number of geometric properties are used to strengthen it.  Data sets of 100 points are easily solved to optimality and the lower bounds provided by the model can be computed up to 300 points.
\end{abstract}

Let $P$ be a set of points in the plane and $n = |P|$ the number of points. 
Let's also $H(P)$ denote the convex hull of $P$ and $I(P)$ the set of internal points of $P$ \emph{i.e} the subset of points that are not vertices of the convex hull $H(P)$. A simple polygon is empty if it does not contain a point of $P$ in its interior. A convex partition of $P$ is a planar subdivision of $H(P)$ into non overlapping empty convex polygons whose vertices are the points of $P$. The minimum convex partitions (MPC) problem is to find the convex partition minimizing the total number of empty convex polygons. In the remaining of the paper, we refer to such empty convex polygons as convex faces or simply \textbf{faces}. Figure \ref{fig_example} illustrates the input data and its convex hull (\ref{fig_instance}) and two feasible solutions figures (\ref{fig_sol}) and (\ref{fig_opt_sol}). Solution (\ref{fig_opt_sol}) is using 5 faces which is optimal for this input data $P$.  A practical application of this problem is reported in the area of network design \cite{Knauer2006ApproximationAF}. Additionally, the authors of \cite{BarbozaSR19} argue that decomposition into polygons plays a key role in algorithm design where divide and conquer schemes take advantage of such decomposition to reduce the problem's size.\\  

\begin{figure}
     \centering
     \begin{subfigure}[b]{0.25\textwidth}
         \centering
         \includegraphics[width=\textwidth]{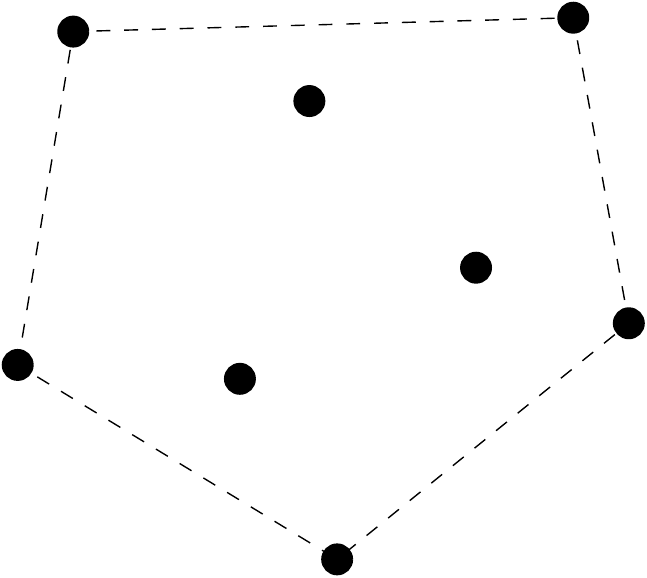}
         \caption{Point set $P$ with convex hull in dashed line}
         \label{fig_instance}
     \end{subfigure}
     \hfill
     \begin{subfigure}[b]{0.25\textwidth}
         \centering
         \includegraphics[width=\textwidth]{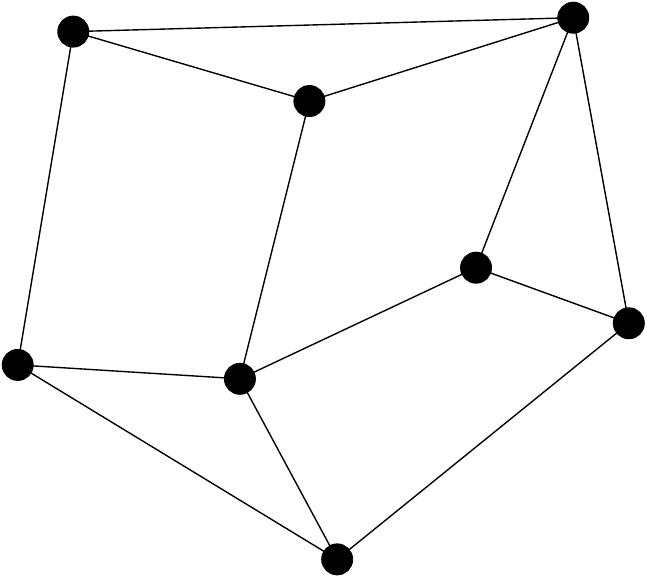}
         \caption{A feasible solution}
         \label{fig_sol}
     \end{subfigure}
     \hfill
     \begin{subfigure}[b]{0.25\textwidth}
         \centering
         \includegraphics[width=\textwidth]{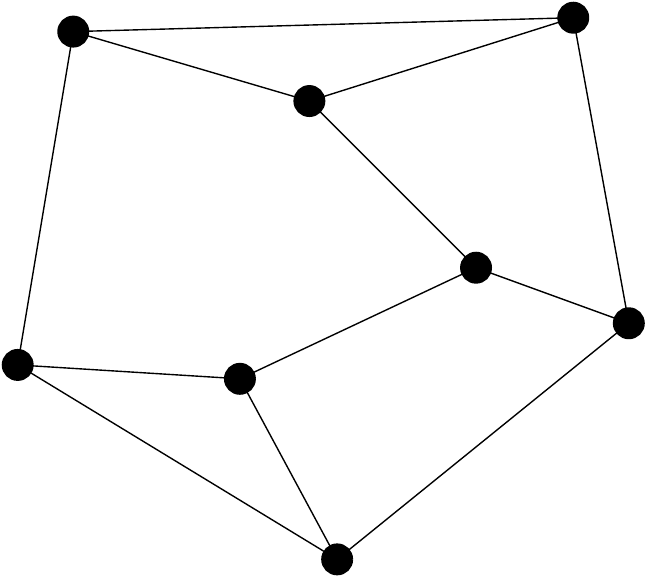}
         \caption{The optimal solution}
         \label{fig_opt_sol}
     \end{subfigure}
        \caption{Example of an instance and two solutions}
        \label{fig_example}
\end{figure}

This problem was the subject of the 2020 Computational Geometry Challenge \cite{demaine2020computing}, which proposed a number of instances of size $n \in [10; 1000000]$ to be solved with no time limit \cite{cg_shop-challenge-page}.  A proof of NP-hardness for the case of planar point sets in not necessarily general position (three points can be collinear) has been announced by N. Grelier in \cite{grelier2020hardness}. The complexity of the MCP for points in general position (no three points are collinear) is still open. If the point sets can be decomposed into a constant number of convex layers, a polynomial-time algorithm is given by Fevens, Meijer and Rappaport in \cite{Fevens2001MinimumCP}.
For instance in general position, Knauer and Spillner gave a 3-approximation algorithm that runs in $O(n\log n)$, and a $\frac{30}{11}-$approximation of complexity $O(n^2)$ in \cite{Knauer2006ApproximationAF}.

The worst-case bound for a set of $n$ points in general position has also been studied, see \cite{HOSONO20091714,Knauer2006ApproximationAF,lomeliharo2012minimal}. J. Urrutia conjecture that this number is at most $n+1$ in \cite{urrutia1998}. At the present time, the best known upper bound is $\frac{4}{3}n-2$ in \cite{sakai2019convex}. Conversely, the best lower bound is $\frac{12}{11}n-2$ in \cite{Lopez2013}.

For the exact resolution of the problem, in addition to the model of Barboza and al. \cite{BarbozaSR19} that we detail below, Da Wei Zheng et al. (winners of the 2020 CG challenge) proposed the use of a SAT model for the exact resolution of instances with $n \le 50$ \cite{zheng_et_al2020}.

We propose a new formulation in integer linear programming for the problem of minimum convex partitioning. It involves an exponential number of variables but $O(n^2)$ constraints. It considerably improves the formulation proposed by Barboza and al. \cite{BarbozaSR19} and allows to entirely close the corresponding benchmark \cite{mcpp-instances-page}.
The key idea of our approach is to reason with convex faces rather edges and points.

Let's first introduce the notations as well as the formulation given in \cite{BarbozaSR19}. Let $\overline{ij}$ the line segment between two distinct points $i, j \in P$ and $S$ the set of all line segments of $P$. Let $E(P)$ the set of edges corresponding to $S$, i.e $E(P) = \{\{i,j\} | \overline{ij} \in S\}$. So the graph $G = (P, E(P))$ is a complete graph induced by $P$. The formulations presented below require to distinguish two \emph{sides} to an edge $\{i,j\}$ namely the \emph{left} and \emph{right} side.
For this, we set an orientation to the segments of $E(P)$ and define $A(P)$ the set of arcs $(i,j)$ such that $\{i,j\} \in E(P)$ and the two points $i$ and $j$ are sorted by x-coordinate i.e $x_i < x_j$ then y-coordinate in case of a tie. Then we can define the orientation of another point $k \in P$ with respect to the arc $(i,j)$ by the cross product: $k$ is said to be to the right of $(i,j)$ if $\vec{ij} \times \vec{ik} \geq 0$ (or $k \in CCW(ij)$), and to the left if $\vec{ij} \times \vec{ik} <0$ ($k \in CW(ij)$). Thereafter, we will sometimes refer either to the edge $\{i,j\}$ or to the arc $(i,j)$, notice that the two sets are in bijection.
A face $f$ can be seen as a sub-graph $G_f=(V_f, E_f)$ of $G$ and we refer to $V_f$ and $E_f$ respectively as the vertices and edges of the face $f$.
Consider a convex face $f$ and an edge $\{i,j\} \in E_f$, $f$ is said to be to the right (resp. left) of the arc $(i,j) \in A(P)$ if for any vertex $k \in V_f$, $k$ is at the right (resp. left) of $(i,j)$. Since $f$ is convex, every vertex of $V_f$ share the same orientation with respect to $(i, j)$. Finally, for each edge $\{i,j\}$, we define the set $R(i,j)$ (resp. $L(i,j)$) the set of all convex faces to the right (resp. left) of arc $(i,j)$.
We denote by $F$ the set of all possible empty convex faces for $P$, see Fig. \ref{enum_conv} for an example.

\begin{figure}
     \centering
     \includegraphics[scale=0.7]{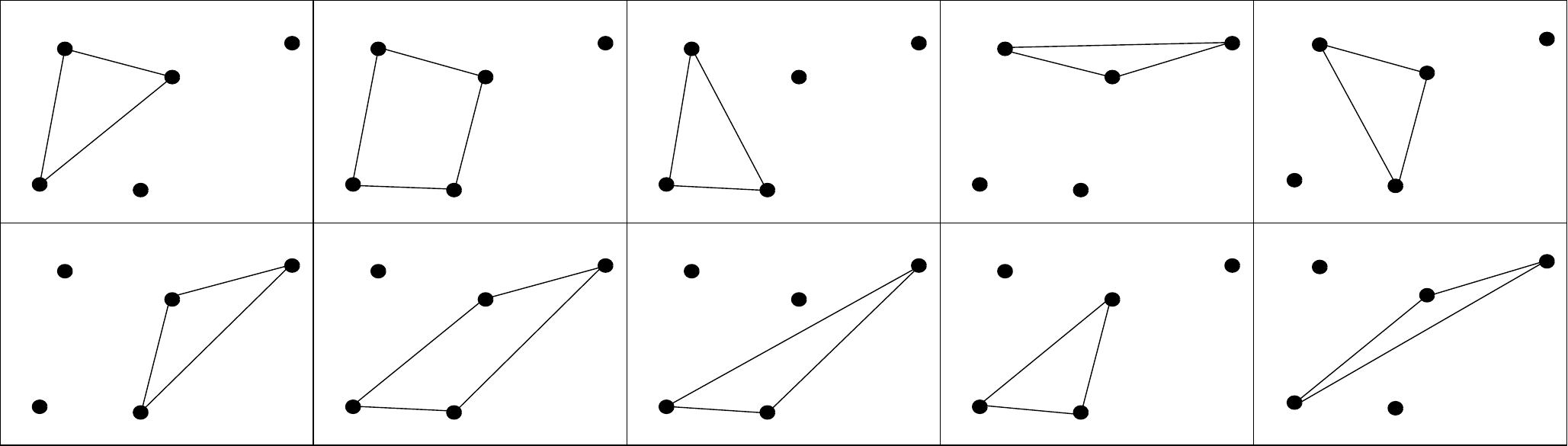}
        \caption{All possible faces $F$ for an instance of $P$ with $n=5$}
        \label{enum_conv}
\end{figure}
 
Let us first recall the formulation proposed by Barboza and al. \cite{BarbozaSR19}. This formulation addresses the problem where the set $P$ is in a general position, i.e. with no three points being collinear. It is a compact formulation with a boolean variable per edge which indicates if this edge is selected in the partition.
Let $S^C \subseteq S$ the set of pairs of crossing line segment and $E^C \subseteq E(P)$ the associated edges.

\begin{center} \begin{tabular}{cc}  & $
\begin{array}{llclllr|}
z^* = \textrm{minimize}& \sum\limits_{(ij) \in E(P)} x_{ij} & & (1)\\
   \mbox{s.t.}& x_{ij} + x_{kl} \leq 1 & \forall \{\{i,j\}, \{k,l\}\} \in E^C & (2) \\
   & x_{ij} = 1 & \forall \{i,j\} \in H(P) & (3) \\
        & \sum\limits_{k \in L(i,j) \cap P} x_{ik} \geq 1 & \forall (i,j) \in A(P), i \in I(P) & (4) \\
        & \sum\limits_{j \in P} x_{ij} \geq 3 & \forall i \in P & (5) \\
        & x_{ij} \in \{0,1\} & \forall \{i,j\} \in E(P) & (6) \\
\end{array} $\\ 
& \\
\end{tabular} \end{center}

The objective function (1) is expressed in terms of number of edges, which is equivalent to the expression in terms of number of faces for this problem by Euler's formula.
Planarity is ensured by the constraint (2). Constraint (3) imposes the selection of the edges of the convex hull $H(P)$ in the solution. Constraint (4) ensures the respect of the local convexity for each internal point $i \in I(P)$, by forcing the selection of an edge to the left of each possible arc $(i,j)$. The last constraint (5) is redundant with the previous convexity constraint (4), and is there to strengthen the formulation and improve the linear relaxation of the model. 

A benchmark is introduced in  \cite{BarbozaSR19} with two sets of instances, one with a number of points $n \le 50$ (generated by keeping only instances for which the model above is able to prove optimality in 20 minutes in order to have a known set of optimal solutions), and the other with a larger number of points: $55 \le n \le 110$.

The paper is organized as follows. Section \ref{geomsection} presents a number of geometric results underlying the model proposed. Section 2 gives the novel integer programming formulation and Section 3 reports some experimental results. 


\section{Geometric observations related to convex faces}
\label{geomsection}
We review a number of observations that will help to establish and strengthen the formulation proposed Section \ref{formulation}. Note first that a very simple lower bound on the number of faces can be obtained using Euler's result. Since the edges of a convex partition form a planar graph, the well-known Euler formula holds and state that $\mathrm{f} = \mathrm{e} - \mathrm{v} + 1$ (without the external face) in any convex partition ($\mathrm{f}$ is the number of faces, $\mathrm{e}$ the number of edges and $\mathrm{v}$ is the number of vertices). Moreover, in a convex partition, the convexity constraint requires that the angles between the incident edges of an internal point be less than or equal to 180 degrees. So the internal points can have a degree two in the collinear case, or at least three if they are not collinear with two other points. The points of $H(P)$ have a degree of at least two with their neighbors in the convex hull. Some vertices can be shown to have a greater degree (see Section \ref{genericcut} and Corollary \ref{hourglass_H}) and for each point $p \in P$, we can consider a lower bound $d(p)$ on its degree. Using Euler's formula, we can compute the following lower bound:

$$\mathrm{f} \ge \frac{1}{2}( \sum_{p \in P} d(p)) - n + 1$$

We refer to this lower bound as the Euler bound.


\subsection{Geometric cuts}
\label{genericcut}
In general position, Knauer and al. in \cite{Knauer2006ApproximationAF} proposed an interesting lower bound on the degrees of the vertices of $H(P)$. They first compute the convex hull of the internal points $H(I(P))$. According to the geometrical configuration of these points, they determine whether one or two edges are needed to connect them to the vertices of the convex hull $H(P)$. In addition, they present examples showing that this bound is almost tight.

We propose a generalization of this bound.
Let $\mathscr{C}$ a convex subset of $\mathbb{R}^2$, $P^{\mathscr{C}}$ the points of $P$ inside $\mathscr{C}$ and $\overline{P^{\mathscr{C}}} = P \setminus P^{\mathscr{C}}$ its complement. We will determine $d(P^{\mathscr{C}})$ a lower bound of the number of edges between $P^{\mathscr{C}}$ and $\overline{P^{\mathscr{C}}}$.

\begin{figure}[h]
     \centering
     \includegraphics[scale=0.7]{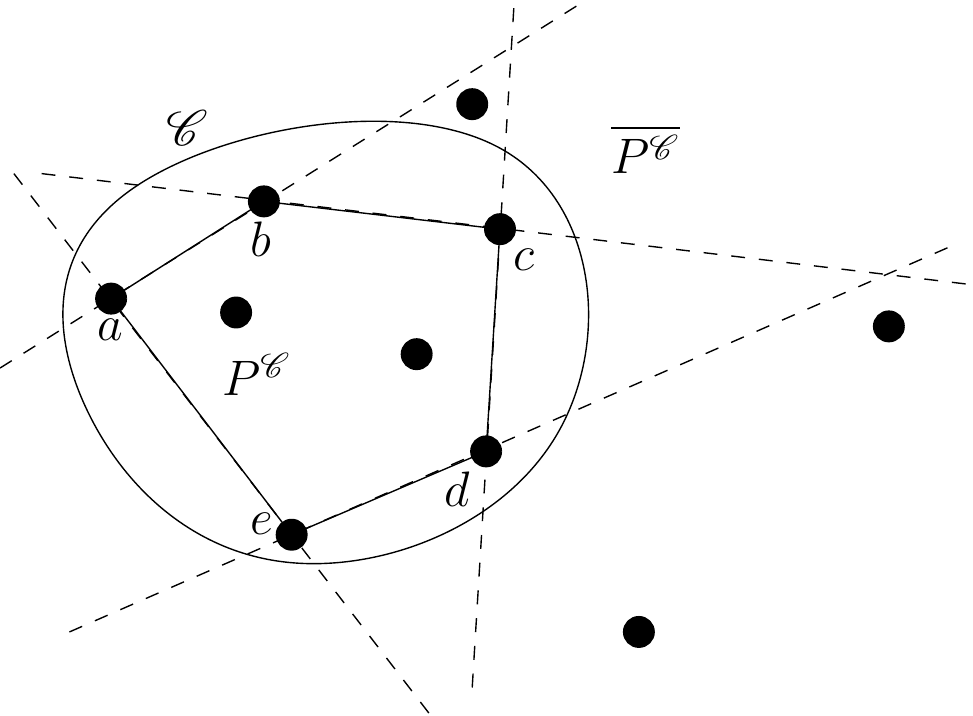}
        \caption{An example of partition with a convex $\mathscr{C}$}
        \label{fig_cut_gene}
\end{figure}

Let $P'^{\mathscr{C}}$ the subset of point $p \in H(P^{\mathscr{C}})$ such that there exists a line segment $\overline{pp'}$ with a point $p' \in \overline{P^{\mathscr{C}}}$ that does not cross $H(P^{\mathscr{C}})$, i.e. such that $\overline{pp'} \cap H(P^{\mathscr{C}}) = p$. In terms of visibility graph, if we consider that $H(P^{\mathscr{C}})$ is an obstacle, for each point $p \in P'^{\mathscr{C}}$, there exists an edge in the visibility graph between $p$ and a point $p' \in \overline{P^{\mathscr{C}}}$.

To deal with collinearity, we now remove from $H(P^{\mathscr{C}})$ the vertices of the convex hull which are collinear with their two neighbors. As in \cite{Knauer2006ApproximationAF}, we classify the vertices of $P'^{\mathscr{C}}$ into distinct sets. 
Let $v$ be a point of $P'^{\mathscr{C}}$ and $u$ and $w$ its neighbors in the convex hull $H(P^{\mathscr{C}})$. Let $\mathcal{H}_w$ (resp. $\mathcal{H}_u$) be the half-plane bounded by the straight line passing through $v$ and $u$ (resp. $w$) and not containing $w$ (resp. $u$). The vertex $v$ is of type (1) if $\mathcal{H}_w \cap \mathcal{H}_u$ contains at least one point of $\overline{P^{\mathscr{C}}}$, type (2) if $v \notin H(P)$ and $\mathcal{H}_w \cap \mathcal{H}_u$ contains no point of $\overline{P^{\mathscr{C}}}$, and type (3) if $v \in H(P)$.
Let $n_1, n_2$ and $n_3$ respectively the number of points of type (1), (2) and (3). Fig.\ref{fig_cut_gene} shows an example of partition with a convex $\mathscr{C}$. The points $\{a, b, c, d, e\}$ belong to the convex hull $H(P^{\mathscr{C}})$.
Point $a \notin P'^{\mathscr{C}}$ because there is no line segment between $a$ and a point of $\overline{P^{\mathscr{C}}}$. So $P'^{\mathscr{C}}$ is the set $\{b, c, d, e\}$.
The point $d$ is of type (1), $c$ is of type (2), and $b, e$ are of type (3) because they belong to $H(P)$.

\begin{lemma}
In any solution, the number of edges between $P^{\mathscr{C}}$ and $\overline{P^{\mathscr{C}}}$ is at least  $n_1 + 2 n_2 + n_3$, $d(P^{\mathscr{C}}) \ge n_1 + 2 n_2 + n_3$.
\label{cut_gene}
\end{lemma}

\begin{proof}
A vertex $v \in P'^{\mathscr{C}}$ have two neighbors in $H(P^{\mathscr{C}})$ and the triplet forms an angle of more than 180 degrees. Since any other point in $P^{\mathscr{C}}$ can satisfy the convexity constraint, in any solution, there must be at least an edge with a point of $\overline{P^{\mathscr{C}}}$. For type (1), there is a point $p' \in \overline{P^{\mathscr{C}}}$ in $H_w \cap H_u$ such that the edge $\{v,p'\}$ can satisfy the convexity constraint. For type (2), at least two edges are needed. Vertices of $H(P)$ (type (3)) don't have to satisfy the convexity constraint, and need only the edge of $H(P)$ between $P^{\mathscr{C}}$ and $\overline{P^{\mathscr{C}}}$, which are mandatory.
\end{proof}

A simple way to divide $P$ into two convex partitions is to use a straight line to partition the points of $P$. We propose a lower bound on the minimum number of edges that cross a straight line in any solution.


\begin{corollary}
Let $l$ a straight line that divides $P$ into two part $P_1$ and $P_2$, and let $H(P_1)$ and $H(P_2)$ their respective convex hulls (see Fig. \ref{hourglass_cuts}).
Let $P_1^l$ (resp. $P_2^l$) the set of point $p \in H(P_1)$ (resp. $H(P_2)$) such that there exists a line segment with a point $p' \in P_2$ (resp. $P_1$) and $\overline{pp'}$ does not cross $H(P_1)$ (resp. $H(P_2)$), i.e. $\overline{pp'} \cap H(P_1) = p$ (resp. $\overline{pp'} \cap H(P_2) = p$). In any solution, there is at least $\max(d(P_1^l), d(P_2^l))$ edges that cross the line $l$.
\label{hourglass}
\end{corollary}

 
\begin{proof}
We simply apply the Lemma \ref{cut_gene} with $P_1$ and with $P_2$, and take the maximum of the two lower bounds.
\end{proof}

\begin{figure}
\centering
 \begin{subfigure}[b]{0.48\textwidth}
         \centering
          \includegraphics[scale=0.6]{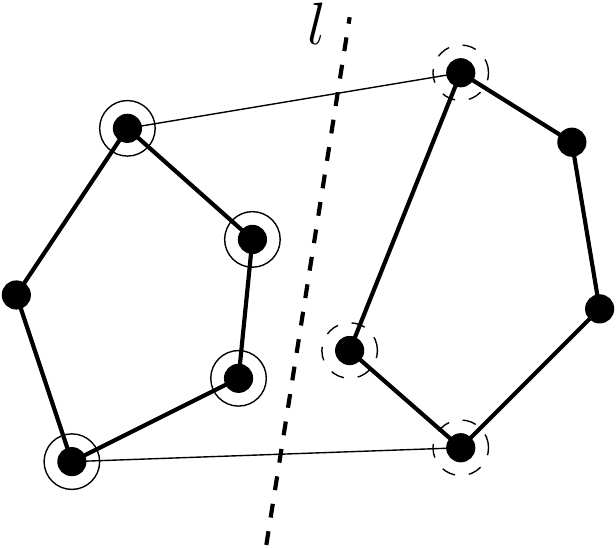}
          \caption{}
          \label{fig_hourglass}
\end{subfigure}
%
\hfill
 \begin{subfigure}[b]{0.48\textwidth}
\centering
      \includegraphics[scale=0.6]{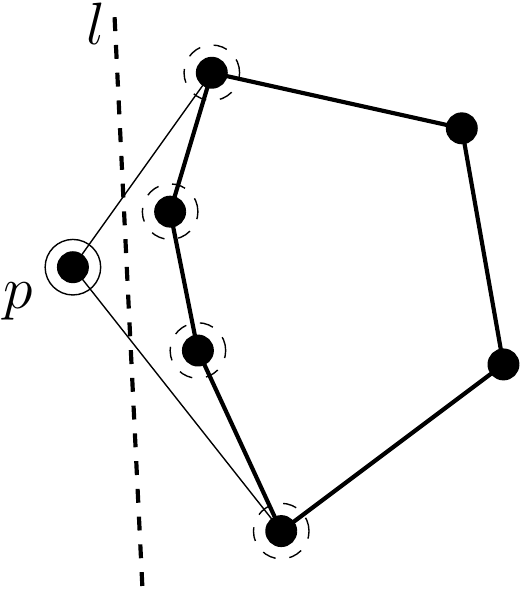}
      \caption{}
      \label{fig_hourglass_p}
\end{subfigure}
   \caption{Two partition of $P$ by a line $l$, $P_1^l$ is circled solid and $P_2^l$ is circled dashed}
   \label{hourglass_cuts}

\end{figure}


\begin{corollary}
\label{hourglass_H}
With Corollary \ref{hourglass}, we obtain a minimum bound on the degree of the points belonging to the convex hull $H(P)$. Indeed, by definition of the convex hull, for a point $p \in H(P)$ we can always partition $P$ with a straight line into two sets $\{p\}$ and $\{P \setminus p\}$, so $d(p)$ is obtained by applying the Lemma \ref{cut_gene} on $P\setminus \{p\}$.
\end{corollary}

Fig. \ref{fig_hourglass_p} show an example of this lower bound on the degree of point in $H(P)$. The leftmost point $p$ has at least 4 edges in any solution, so $d(p) = 4$.


\subsection{Mandatory face}


Recall that F is the set of all faces of $P$.
Among this set, some convex face must necessarily be taken in any optimal solution. We can therefore directly select them in the solution.

\begin{remark}
\label{mandatory_faces}
If a point of $\mathbb{R}^2$ inside the convex hull $H(P)$ is covered by only one convex $f \in F$, $f$ must be in the solution.
\end{remark}

Indeed, since any point of $\mathbb{R}^2$ inside the convex hull $H(P)$ must belong to a face of the partition solution, the only face $f$ that contains this point must be part of the solution.

An example of a mandatory face is shown in Fig. \ref{mandatory_edge}. The edge $\{i,j\}$ belongs to the convex hull, so it is mandatory. Let us now observe that taking the line passing through $\{j,k\}$, the side containing the $\{i,j\}$ edge contains no other point than $i$. Similarly, taking the line passing through $\{i, k\}$, the side containing the edge $\{i,j\}$ does not contain any other point that $j$. Therefore, there is no other convex face containing the points in $\mathbb{R}^2 \cap \{i,j\} \setminus \{i\}, \{j\}$ than the face defined by vertices $\{i,j,k\}$.

\begin{figure}[h]
\centering
      \includegraphics[scale=0.6]{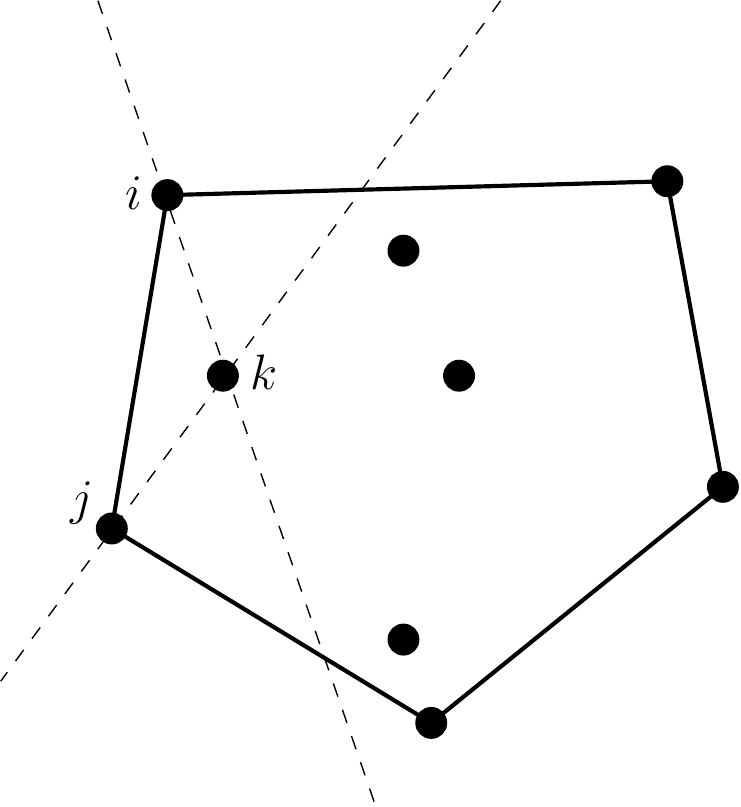}
   \caption{Mandatory face $\{i, j, k\}$}
   \label{mandatory_edge}
\end{figure}


\subsection{Dominated face}
Conversely, among the set $F$, one can determine a set of dominant polygons which are sufficient to obtain an optimal solution. We are now going to describe several sets of faces that can be eliminated from the faces to be considered without losing any optimal solution. 

Two faces $f_1$ and $f_2$ are said to be adjacent along an edge $e$ if $E_{f_1} \cap E_{f_2} = e$. The union $f_1 \cup f_2$ of two adjacent faces $f_1$ and $f_2$ is a polygon defined by the union of their surfaces and the removal of the common edge $e$ i.e  $E_{f_3} = E_{f_1} \cup E_{f_2} \setminus \{e\}$. The resulting polygon $f_3$ is not necessarily convex.

\begin{lemma}
\label{domi_face}
A face $f$ is dominated if for all $f'$ adjacent to $f$, $f \cup f'$ is convex (is in $F$).
\end{lemma}

\begin{corollary}
\label{domi_edge}
If an edge belongs only to dominated faces, this edge is itself dominated and does not belong to an optimal solution.
\end{corollary}

For instance, an edge of two non-consecutive vertices of the convex hull is a dominated edge.
The points of the convex hull locally satisfy their convexity constraint with the edges of the convex hull which are mandatory. So the union between two faces that share an edge between two vertices of the convex hull is always convex and a solution without this edge is always better (see Fig. \ref{dom_edge}).

\begin{figure}
\centering
      \includegraphics[scale=0.8]{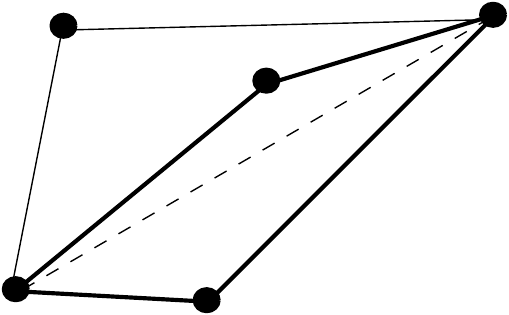}
   \caption{From example of Fig. \ref{enum_conv}, the dashed edge is dominated, as well as the two triangle faces that contain it}
   \label{dom_edge}
\end{figure}


 



\subsection{Generation of the convex faces set}
\label{generationoffaces}
In section \ref{formulation}, our integer linear program requires the computation of $F$ \emph{i.e} the enumeration of the possible faces for $P$. See Fig.~\ref{enum_conv} for an example of set $F$ ($n = 5$). We rely on the article of Dobkin and al. \cite{DobkinEO90} which describes an algorithm to find the set of empty convex $r$-gons ($r$ is the number of points of the polygon). The principle of the algorithm is as follows. For each point $p \in P$, a star-shaped polygon is formed by sorting the vertices to the left of $p$ by angle around $p$. We then compute the visibility graph of this polygon, including the edges of the polygon, but removing the edges containing $p$. We finally compute the set of convex chains in the visibility graph (in \cite{DobkinEO90} the algorithm lists the convex chains of r-2 edges, but here we want all the faces of any size). For each chain, by adding the two edges connecting the ends with $p$, we obtain a face. 

With the properties seen in this section, we can preprocess the $F$ set to eliminate a subset of dominated faces, we call this subset $F^{P\!P} \subseteq F$ (PP for PreProcessing).
In practice, filtering all the dominated faces can be very expensive. We make a preprocessing to filter a part of the dominated faces in a polynomial time.

In order to filter the dominated faces, we use the following algorithm:
\begin{algorithm}
\caption{Faces pre-processing \label{AlgoFPP}}
\begin{algorithmic}[1]
\STATE $F^{P\!P} \leftarrow F$
\FOR{each edge $\{i,j\}$} 
\IF{$\min(|R(i,j)|, |L(i,j)|) \leq 3$} \label{limit3}
\IF {the union of any two faces of $R(i,j) \times L(i,j)$ is convex}
\STATE Remove all faces using edge $\{i,j\}$ from $F^{P\!P}$.
\ENDIF
\ENDIF
\ENDFOR
\FOR{each face $f \in F^{P\!P}$ and each edge $\{i,j\}$ of $E_f$} 
\IF{the union of $f$ with all its adjacent faces along $\{i,j\}$ on the other side (if $f \in R(i,j)$, all faces of $L(i,j)$, otherwise all faces of $R(i,j)$) is convex}
\STATE Remove the face $f$ from $F^{P\!P}$.
\ENDIF
\ENDFOR
\FOR{each edge $\{i,j\}$ s.t $i$ and $j$ are non-consecutive vertices of the convex hull $H(P)$}
\STATE Remove all faces using edge $\{i,j\}$ from $F^{P\!P}$.
\ENDFOR
\end{algorithmic}
\end{algorithm}

Note that the condition at line \ref{limit3} deliberately limits the convexity verification of the union of $R(i,j) \times L(i,j)$ faces to cases where $|R(i,j)|$ or $|R(i,j)|$ is less than three for efficiency reasons. Note also that we are only deleting faces. Indeed, the starting set $F$ is made by definition of all the possible faces so when we check if the union of two faces is convex, the resulting face of this union is already in $F$.
\section{An IP formulation based on convex polygons}

\subsection{Formulation}
\label{formulation}
Recall that $F$ denotes the set of all convex faces. Consider a variable $x_f \in \{0,1\}$ for each face $f\in F$ that is set to 1 if $f$ is in the convex partition, and 0 otherwise. 


We propose the following Integer Programming model:

\begin{center} \begin{tabular}{cc}  & $
\begin{array}{llllllr|}
&\textrm{minimize } z =& \sum_{f \in F} x_f & & (1)\\
        &  (M) & \sum_{f \in L(i,j)} x_f - \sum_{f \in R(i,j)} x_{f} = & \left\{
    \begin{array}{ll}
       1  & \forall $\{i,j\}$ \in H(P), \:L(i,j) \neq \emptyset \\
        -1  & \forall $\{i,j\}$ \in H(P), \: R(i,j) \neq \emptyset \\
        0 & \forall $\{i,j\}$ \in E \setminus H(P)
    \end{array}
\right. & (2) \\
      &  & x_f \in \{0, 1\} & \forall f \in F & (3) \\
       
\end{array} $\\ 
& \\
\end{tabular} \end{center}

The objective function (1) minimize the total number of faces chosen for the solution. Constraints (2) ensure that for each edge of $H(P)$, there is a face chosen at the left or right, depending on the orientation of the edge, in the solution. It also ensures that for all remaining internal edges, the number of faces chosen on its right is equal to the number of faces chosen on its left (this number will be either 0 or 1). The constraint is very similar to a flow or path conservation in a network and can be understood as a \emph{conservation of faces} \emph{i.e} that whenever a face is used on one side of an edge, another face must match the other side (to the exception of hull which acts as a \emph{source} or \emph{sink} of faces). It enforces the faces chosen to tile the interior of $H(P)$. Note that the model is however not a network flow model.

Although this formulation is correct as it is, it can be strengthened by enforcing the minimum number of convex faces, denoted $d(i)$ that are adjacent to each point $i \in P$ (the degree of a vertex). Note $d(i) \geq 2$ using the reasonings presented Section \ref{genericcut}:
$$\sum_{f| i \in V_{f}} x_f \ge d(i) \quad \forall i \in P \qquad (4)$$
Finally note that the linear relaxation consists in relaxing constraint $(3)$ into $x_f \geq 0$ and the optimal value is denoted $z^*_{LP}$. \\

\noindent \textbf{Cutting planes.} Observe that a feasible integer solution of $(M)$ is an independent set in the intersection graph of the convex faces. More precisely, let $I = (V_F, E_F)$ be an undirected graph where each vertex $v_f \in V_F$ is associated to a convex face of $f \in F$ and an edge $(v_{a}, v_{b})$ is added to $E_F$ if faces $a$ and $b$ intersect. Any feasible solution to $(M)$ is an independent set of $I$ although the converse is not true. As a result, a number of valid inequalities known for independent set can be used to strengthen $(M)$. Since the size of $I$ is exponential in $n$, such inequalities must be added as cutting planes by solving a separation problem in the solution of the linear relaxation. We considered two classes of such inequalities, the clique and the odd cycle inequalities. Let $C \subseteq V_F$ be a clique of $I$ \emph{i.e} a set of two by two intersecting faces. Since only one face of $C$ can be chosen in a solution, the following cutting plane can be added:
$$\sum_{f \in C} x_f \le 1$$
Let $O \subseteq V_F$ be an odd cycle of $I$, the following inequality holds:
$$\sum_{f \in C} x_f \le \frac{1}{2}(|O|-1) $$

\noindent \textbf{Size of the formulation and preprocessing.} The enumeration of all convex faces required by $(M)$ is discussed Section \ref{generationoffaces}. It is easy to see that the number of convex faces grows exponentially with the number of points. The worst case is hit when all points belongs to the convex hull. In this case, any subset of points define a valid convex face leading to $2^n$ convex faces. In general, $|F| \in O(2^n)$ and formulation $(M)$ have an exponential number of variables. In practice, some points lies inside the hull and not all subsets of points define a convex face. We will see Section \ref{expe} that the formulation scales very easily to practical instances of size $n=100$. After preprocessing of the faces, the formulation is stated on the set $F^{P\!P}$ and the mandatory faces $f$ (See Remark \ref{mandatory_faces}) are enforced with $x_f = 1$. The number of constraints is in $O(n^2)$ and the preprocessing also helps reducing this number in practice (See Corollary \ref{domi_edge}).




\section{Experimental results}
\label{expe}

We report the experimental evaluation of the proposed model.	Let's first give details on the hardware, software as well as the benchmark used.

\textbf{Computational Environment.} All experiments were run on a Macbook pro with 4 processors Intel Core i7 at 2.8 GHz and a limit of 4 Go of RAM. Algorithms are all implemented in Java and CPLEX 12.8.0 is used in multi-thread mode. Note that experiments of Barboza and al \cite{BarbozaSR19}, which serve as a baseline, were run in single thread mode.

\textbf{Benchmark.} Two sets of instances $T1$ and $T2$ are available from \cite{mcpp-instances-page}. $T1$ is made of 30 instances for each size $n$ in $30,32, \ldots, 50$ that is a total of $330$ instances. Each instance is generated using a uniform random distribution of the coordinates in the interval $[0,1]$. But for each size, the first 30 instances that were found optimally solvable within 20 minutes were kept. Some instances were rejected (hit the time limit) from size $n=44$ and onwards so that set $T1$ is biased to be easy enough for the model given in \cite{BarbozaSR19}. Set $T2$ is made of 30 instances for each sizes $55, 60, \ldots, 110$ (so 330 in total) and no selection was made on set $T2$.\\
A set of more \emph{structured} instances was proposed in the CG:SHOP challenge 2020  \cite{cg_shop-challenge-page, demaine2020computing}. We use instances of sizes $\leq 100$ from the image set. This leads to four instances (four images: euro-night, london, stars, us-night) for each size $n$ in $10,15,20,\ldots,100$ so a total of 56 instances. We also report some results on the 20 instances of size $100 < n \leq 300$ of the challenge\footnote{All the instances used as well as the detailed results are available on the webpage \url{https://pagesperso.g-scop.grenoble-inp.fr/~cambazah/convexpartition/minconvex.html}}.

The results are presented in two parts. Firstly, we evaluate our model by solving the integer model. Secondly, we discuss the lower bound that can be obtained for larger problem of sizes $150 < n \leq 300$ by focusing on the linear relaxation.

\subsection{Exact solving}
\begin{table}[t]
\footnotesize
\begin{tabular}{|r|r|r|r|r|r|r|r|r|r|r|r|r|r|r|r|r|}
\hline
 \# inst &	n	&	\multicolumn{2} {  c | }{Faces}	&	\multicolumn{3} {c | }{\%Gap}	&	\multicolumn{3} {  c | }{Cpu(s)}&	\#nodes \\
\hline
           & 		&	Avg $|F|$	 		& Avg $|F^{P\!P}|$					&	Med   & Avg 		&	Max		&	Med 	&	Avg	&	Max         &	Max		    \\
\hline
\hline 
330	& 	30-50	&	11417	& 	5592	&	0,00	&	0,02	&	3,85	&	0,35	&	0,41	&	1,58	&	0 \\
\hline  
30	& 	55	&	23167	&	12253	&	0,00	&	0,20	&	3,13	&	0,87	&	1,07	&	1,86	&	0,00 \\
30	& 	60	&	27974	& 	15141	&	0,00	&	0,00	&	0,00	&	1,11	&	1,35	&	2,48	&	0,00 \\
30	& 	65	&	35061	&	19201	&	0,00	&	0,17	&	2,56	&	1,79	&	1,89	&	3,00	&	0,00 \\
30	& 	70	&	40933	& 	22392	&	0,00	&	0,32	&	2,56	&	2,46	&	2,31	&	6,13	&	45,00 \\
30	& 	75	&	48998	&	27818	&	0,00	&	0,44	&	2,38	&	3,55	&	3,91	&	9,37	&	25,00 \\
30	& 	80	&	55737	& 	31606	&	0,00	&	0,41	&	2,08	&	4,09	&	3,71	&	6,61	&	0,00 \\
30	& 	85	&	66366	&	37677	&	0,00	&	0,48	&	2,17	&	5,31	&	5,45	&	13,18	&	26,00 \\
30	& 	90	&	73386	& 	42299	&	0,00	&	0,51	&	1,92	&	5,74	&	7,49	&	35,18	&	77,00 \\
30	& 	95	&	84009	&	48642	&	0,00	&	0,90	&	3,51	&	8,83	&	16,23	&	82,61	&	1913,00 \\
30	& 	100	&	95135	& 	55078	&	0,00	&	0,47	&	1,85	&	8,60	&	13,58	&	62,84	&	125,00 \\
30	& 	110	&	116568	&	69006	&	1,55	&	1,00	&	3,08	&	18,16	&	31,78	&	166,73	&	2399,00 \\
\hline
\end{tabular}
\caption{Experimental results on the entire benchmark T1 and T2 of Barboza and al \cite{BarbozaSR19}. All instances are solved to optimality. \label{barboza} }
\end{table}

Table \ref{barboza} and  \ref{cgshop}  presents the results on two benchmarks. Each line gives a summary of the results for a class of instances gathered by size $n$ and column  \emph{\# inst} gives the number of instances considered in the class. Note that all instances of size $n\leq50$ are considered in a single class (the first line). For each class of instances, we report a number of metrics. We report the average number of faces (\emph{Avg $|F|$}) as well as the number of faces after preprocessing using Algorithm \ref{AlgoFPP} (\emph{Avg $|F^{P\!P}|$}). We solve the linear relaxation $z^*_{LP}$ of our model \footnote{Note that this is the linear relaxation and not the root node lower bound that is usually much stronger after the automatic preprocessing performed by cplex.} (changing all $x_f \in \{0, 1\}$ into $x_f \geq 0$) and reports the median (\emph{Med}), average (\emph{Avg}) and maximum (\emph{Max}) gap (column \emph{\%Gap}) computed as follows: $100(z^* -  \lceil z^*_{LP} \rceil)/z^*$. We also report the median, average  and maximum solving time (\emph{Cpu(s)}) in seconds. Finally, the maximum number of nodes (\emph{\# Nodes}) explored by the branch and bound algorithm is given. 

\begin{itemize}
\item All instances of size $n \leq 50$ are easily solved optimally with a maximum time of 1,58 seconds (Table \ref{barboza}). Note the none of the instances require branching from the solver \emph{i.e} that all are solved at the root node.
\item The approach remains very efficient even for instances of size up to $n=110$ with an average time 31,78 seconds (Table \ref{barboza}). Branching is sometimes required from size $n=75$ and onwards but the quality of the linear relaxation appears to be very good with a maximum gap of 3.51\% across all instances of size $55 \leq n \leq 110$.
\item The preprocessing of faces remove around 44\% of them in average across all instances. However it seems to decrease slowly as the size increases and is around 41\% for instances of size $n=100$ (Table \ref{barboza}).
\item The results observed on the structured image instances (Table \ref{cgshop}) of the challenge are very similar and all instances are easily solved to optimality.
\item Table \ref{cgortho} gives the results for instances with numerous collinear points on vertical and horizontal lines. Although these structured instances tend to have a large number of faces, the model can solve optimally the three instances available with a size around 100. Considering the large number of faces involved, we also give the results without pre-processing of the faces to show its interest.
\end{itemize}

Even though the solving times can not be directly compared to \cite{BarbozaSR19} due to different hardware and number of threads, we believe it is a significant improvement. Some of the instances of size $\leq 50$ could require a computation time of 20 minutes and none of instances with $55 \leq n \leq 100$ can be solved exactly in \cite{BarbozaSR19} where elaborate heuristics are used instead.
\begin{table}
\footnotesize
\begin{tabular}{|r|r|r|r|r|r|r|r|r|r|r|r|r|r|r|r|r|}
\hline
 \# inst &	n	&	\multicolumn{2} {  c | }{Faces}	&	\multicolumn{3} {c | }{\%Gap}	&	\multicolumn{3} {  c | }{Cpu(s)}&	\#nodes \\
\hline 
           & 		&	Avg $|F|$	 & Avg $|F^{P\!P}|$		&	Med.   & Avg 		&	Max		&	Med 	&	Avg	&	Max         &	Max		    \\
           \hline
           \hline
           4	&	10	&	141	&	46	&	0,00	&	0,00	&	0,00	&	0,00	&	0,00	&	0,01	&	0\\
4	&	15	&	532	&	186	&	0,00	&	0,00	&	0,00	&	0,01	&	0,01	&	0,03	&	0\\
4	&	20	&	1308	&	659	&	0,00	&	0,00	&	0,00	&	0,03	&	0,04	&	0,07	&	0\\
4	&	25	&	2717	&	1276	&	0,00	&	0,00	&	0,00	&	0,06	&	0,06	&	0,06	&	0\\
4	&	30	&	3921	&	2170	&	0,00	&	0,00	&	0,00	&	0,11	&	0,11	&	0,15	&	0\\
4	&	35	&	6455	&	4047	&	0,00	&	0,00	&	0,00	&	0,15	&	0,14	&	0,16	&	0\\
4	&	40	&	11080	&	5627	&	0,00	&	0,00	&	0,00	&	0,36	&	0,38	&	0,61	&	0\\
4	&	45	&	12989	&	6702	&	1,79	&	1,82	&	3,70	&	0,42	&	0,50	&	0,86	&	0\\
4	&	50	&	20686	&	11826	&	0,00	&	0,00	&	0,00	&	1,18	&	1,42	&	2,44	&	0\\
4	&	60	&	26780	&	16091	&	0,00	&	0,00	&	0,00	&	1,77	&	1,63	&	2,25	&	0\\
4	&	70	&	39468	&	24258	&	0,00	&	0,00	&	0,00	&	1,35	&	1,58	&	2,44	&	0\\
4	&	80	&	53071	&	34618	&	1,00	&	1,04	&	2,17	&	5,41	&	5,83	&	8,86	&	0\\
4	&	90	&	76710	&	47341	&	0,00	&	0,00	&	0,00	&	7,26	&	6,40	&	7,79	&	0\\
4	&	100	&	91813	&	53815	&	0,00	&	0,43	&	1,72	&	11,53	&	17,55	&	42,53	&	63\\
           
\hline


\end{tabular}
\caption{Experimental results on the images benchmark ($n\leq 100$) of the challenge. All instances are solved to optimality. Each class is made of four instances referred to as euro-night, london, stars and us-night. \label{cgshop}}
\end{table}

\begin{table}[h]
\footnotesize
\begin{tabular}{|r|r|r|r|r|r|r|r|r|r|r|r|r|r|r|r|r|}

\hline         
name		&	n	&	$|F|$		&$|F^{P\!P}|$	& Pp(s)	& Euler	&$ \lceil z^*_{LP} \rceil$	&	$z^*$	&\%Gap	&	Cpu(s)	&\#nodes	\\
\hline
\hline
rop0000101	&	101	&	449341	&	326348	& 16,43 & 	4	&	21	&	21	&	0	&	177,14	&	10 \\
rop0000107	&	107	&	785030	&	620263	& 262,03 & 	4	&	23	&	24	&	4,17	&	1195,59	&	80 \\
rop0000122	&	122	&	703772	&	553370	& 30,68 & 	3	&	22	&	23	&	4,35	&	748,73	&	109 \\
\hline
\hline
rop0000101	&	101	&	449341	&	- 	&		0,39	&	4	&	21	&	21	&	0	&	222,53	&	14 \\
rop0000107	&	107	&	785030	&	- 	&		0,6	&	4	&	23	&	24	&	4,17	&		1389,5	&	90 \\
rop0000122	&	122	&	703772	&	- 	&	0,81	&	3	&	22	&	23	&	4,35	&	1495,05	&	67 \\
\hline
\end{tabular}
\caption{Experimental results on the instances ($ 100< n < 150$) of the challenge with preprocessing (first three lines) and without preprocessing (three last lines). The three instances are solved to optimality.\label{cgortho}}
\end{table}

We now turn our attention to larger instances and discuss the possibility to use model $(M)$ to produce a lower bound.

\subsection{Lower bounds}
We now consider the linear relaxation $z^*_{LP}$ of model $(M)$ on larger instances. Table \ref{cglargeshop} reports the results obtained on all instances of the challenge \cite{cg_shop-challenge-page} of sizes $150 < n \leq 300$\footnote{
us-night-0000200,
paris\-0000200,
stars\-0000200,
uniform\-0000200\-1,
uniform\-0000200\-2,
euro\-night\-0000200,
ortho\_rect\_union\_170,
ortho\_rect\_union\_186,
ortho\_rect\_union\_199,
ortho\_rect\_union\_208,
rop0000262,
us-night\-0000300,
paris\-0000300,
stars\-0000300,
uniform\-0000300-1,
uniform\-0000300-2,
euro\-night\-0000300}. 
For each instance, the table gives the best known upper bound obtained during the challenge (\emph{UB}), the number of convex faces remaining after preprocessing ($|F^{P\!P}|$), the percentage of reduction compared to the initial number of faces (\emph{\%RF}) computed as $100(1-\frac{|F^{P\!P}|}{|F|})$, the percentage of reduction of the number of edges (\emph{\%RE}) computed as $100(1-\frac{|E^{P\!P}|}{n(n-1)/2})$ the cpu time in seconds required for pre-processing (\emph{Pp(s)}), the value of the Euler lower bound (\emph{Euler}), the value of the linear relaxation of the proposed model ($\lceil z^*_{LP} \rceil$), the gap to the best known upper bound (\emph{\%Gap}) computed as $100(U\!B -  \lceil z^*_{LP} \rceil)/UB$, the time for computing the relaxation (\emph{Slv(s)}) in seconds and the total time in seconds (\emph{Tot(s)}).

\begin{itemize}
\item Despite the large number of convex faces, the model scales much better than expected and solving the linear relaxation itself is faster than the preprocessing of faces. 
\item The number of convex faces is considerably larger for the instances referred to as \emph{ortho} or \emph{rop} where the points of $P$ are collinear on vertical and horizontal lines. Note that Euler's formula lead to a very weak bound in this latter case. Note also that more than a third of the edges are often removed by preprocessing.
\item The lower bound computed improves significantly the bound provided by the Euler's formula with gaps less than $3\%$ for the image benchmark and gaps less than $8\%$ for the orthogonal benchmark.
\end{itemize}

Experiments using the cutting planes mentioned remained unconclusive so far and are not reported in details in the present paper. The cutting planes seem to provide small improvements of the bound and the separation routine is computationally expensive.

\begin{table}
\footnotesize
\begin{tabular}{|r|r|r|r|r|r|r|r|r|r|r|r|r|r|r|r|r|r|}
\hline
   &	 &  &	\multicolumn{4} {  c | }{Face generation}	&          &	\multicolumn{4} {c | }{Linear relaxation}	 \\
\hline
   name         & n	   &	UB  & 	$|F^{P\!P}|$	& \%RF	& \%RE & Pp(s)    &Euler  & $ \lceil z^*_{LP} \rceil$ & \%Gap	&	 Slv(s) & Tot(s)  \\
\hline
\hline
us	&	200	&	107	&	298473	&	37,2	&	36,1	&	4,5	&	97	&	105	&	1,87	&	19,7	&	25,4 \\
paris	&	200	&	110	&	259023	&	36,6	&	36,4	&	5,3	&	99	&	108	&	1,82	&	15,9	&	22,2 \\
stars	&	200	&	110	&	285256	&	37,5	&	35,5	&	1,6	&	100	&	109	&	0,91	&	23,4	&	26,0 \\
unif-1	&	200	&	112	&	271086	&	36,9	&	36,4	&	2,3	&	100	&	110	&	1,79	&	16,5	&	19,8 \\
unif-2	&	200	&	111	&	280183	&	37,7	&	36,9	&	3,0	&	98	&	109	&	1,80	&	16,1	&	21,8 \\
euro	&	200	&	110	&	263472	&	35,5	&	35,8	&	2,2	&	98	&	108	&	1,82	&	16,5	&	19,6 \\
ortho	&	170	&	62	&	404356	&	30,1	&	35,6	&	5,4	&	10	&	59	&	4,84	&	10,8	&	17,7 \\
ortho	&	186	&	65	&	469839	&	28,3	&	33,9	&	4,1	&	11	&	60	&	7,69	&	20,3	&	26,0 \\
ortho	&	199	&	71	&	527131	&	30,7	&	34,6	&	11,6	&	11	&	66	&	7,04	&	22,7	&	36,3 \\
ortho	&	208	&	74	&	703891	&	24,4	&	34,4	&	8,3	&	10	&	69	&	6,76	&	25,2	&	36,3 \\
rop	&	262	&	41	&	3225861	&	15,9	&	22,2	&	211,4	&	3	&	40	&	2,44	&	342,3	&	574,6 \\
us	&	300	&	160	&	760641	&	34,2	&	35,4	&	4,9	&	149	&	158	&	1,25	&	88,1	&	95,7 \\
paris	&	300	&	161	&	651331	&	34,8	&	36,3	&	4,7	&	149	&	157	&	2,48	&	82,9	&	90,0 \\
stars	&	300	&	163	&	760654	&	33,8	&	35,4	&	6,6	&	149	&	160	&	1,84	&	101,3	&	110,7 \\
unif-1	&	300	&	161	&	654205	&	36,9	&	36,6	&	7,1	&	143	&	158	&	1,86	&	81,8	&	91,8 \\
unif-2	&	300	&	167	&	640052	&	36,2	&	37,0	&	4,6	&	155	&	163	&	2,40	&	74,4	&	81,5 \\
euro	&	300	&	163	&	776699	&	33,6	&	35,7	&	6,4	&	147	&	158	&	3,07	&	84,8	&	94,0 \\
\hline
\hline
\end{tabular}
\caption{Lower bounds ($ \lceil z^*_{LP} \rceil$) computed on all instances of the challenge with ($ 150 < n \leq 300$). \label{cglargeshop}}
\end{table}


\section{Conclusion and future work}
\label{conclu}

We propose a new formulation in integer linear programming for the minimum convex partition problem. It proves able to solve to optimality all instances of less than a 100 of points proposed so far in the literature, considerably improving the formulation given by \cite{BarbozaSR19}. Despite the exponential number of variables involved, its linear relaxation can be solved efficiently for instances of size up to 300 points providing strong lower bounds.

A first direction of research is to investigate further the cutting planes that could be used from the independent set formulation of the problem since a large family of such inequalities are already known. More interestingly, cutting planes can also be derived from geometric statements. For instance, the result proposed in Lemma \ref{cut_gene} directly lead to a cutting plane but require to study a separation algorithm. A second direction of research is to propose an implicit enumeration of the convex faces focusing on the faces with negative reduced cost only as opposed to generate and preprocess the entire set of faces. In other words, the linear relaxation of $(M)$ could be solved using a column generation procedure.



\bibliography{minConvPart}

\end{document}